\documentclass{article}

\usepackage[margin = 3cm]{geometry}
\usepackage{bm}
\usepackage{amsmath, amsthm}
\usepackage{amsfonts}
\usepackage{mathtools}
\usepackage{graphicx}
\usepackage{subfig}
\usepackage{hyperref}
\usepackage{cleveref}
\usepackage{epstopdf, epsfig}
\usepackage{color}
\usepackage{authblk}

\newtheorem{lemma}{Lemma}
\newtheorem{Theorem}{Theorem}
\newtheorem{Corollary}{Corollary}
\newtheorem{proposition}{Proposition}
\newtheorem{Definition}{Definition}

\newcommand*\diff{\mathop{}\!\mathrm{d}}
\newcommand{\defeq}{\vcentcolon=}

\title{A derivation of the Liouville equation for hard particle dynamics with non-conservative interactions}

\author[1]{B. D. Goddard}
\author[1]{T. D. Hurst\thanks{Corresponding Author: \href{mailto: th19@hw.ac.uk}{t.hurst@sms.ed.ac.uk}.}}
\author[2]{M. Wilkinson}
\affil[1]{School of Mathematics and the Maxwell Institute for Mathematical Sciences, University of
	Edinburgh, Edinburgh, UK, EH9 3FD}
\affil[2]{School of Mathematics, Heriot-Watt University, Edinburgh, UK, EH14
	4AS}

\begin{document}
	\maketitle
	
	\begin{abstract}
		The Liouville equation is of fundamental importance in the derivation of continuum models for physical systems which are approximated by interacting particles. However, when particles undergo instantaneous interactions such as collisions, the derivation of the Liouville equation must be adapted to exclude non-physical particle positions, and include the effect of instantaneous interactions. We present the weak formulation of the Liouville equation for interacting particles with general particle dynamics and interactions, and discuss the results using an example.
	\end{abstract}
	\section{Introduction}\label{section:Intro}
	
	Many physical systems can be interpreted as a collection of interacting particles, for example interactions in and between molecules \cite{Schlick2010}, colloidal systems \cite{Goddard2012}, or systems of granular media \cite{Lun1984,Marini2007,Radjai2011}. However, when considering a large number of particles, simulating such a system as a discrete set of particles quickly becomes computationally intractable. In these cases, it is necessary to consider a continuous approximation of the system. One of the most popular first steps to a valid continuous model is the Liouville equation (when particle dynamics are deterministic) or the Kramers equation \cite{Risken1996} (when particle dynamics are stochastic).
	
	We assume that a system of $N$ particles in $d$ dimensions with positions $X(t)\in\mathbb{R}^{dN}$ and velocities $V(t)\in\mathbb{R}^{dN}$ at time $t\in\mathbb{R}$ are governed by Newton's equations:
	\begin{align}
	\frac{\diff X(t)}{\diff t}=\frac{V(t)}{m},\quad
	\frac{\diff V(t)}{\diff t}=G(X(t),V(t),t),\label{eq:GeneralNewtons}
	\end{align}
	where $G(X,V,t)$ incorporates external effects such as gravity and friction, and interparticle interactions such as cohesion in granular media or intermolecular forces in molecular dynamics. Under the assumption that the microscopic dynamics are smooth, associated with the microscopic dynamics is the Liouville equation, a partial differential equation which determines the dynamics of the $N$-body distribution function $f^{(N)}(X,V,t)$.
	\begin{align}
	\mathcal{L}[f^{(N)}] \defeq \Bigg[\frac{\partial}{\partial t}+\frac{1}{m}v \cdot \nabla_{X}-\nabla_V\cdot G(X,V,t)\Bigg]f^{(N)}(X,V,t)=0.\label{eq:LiouvilleGeneralNoColl}
	\end{align}
	For point-like particles and sufficiently smooth $G$, \cref{eq:LiouvilleGeneralNoColl} fully describes the evolution of an initial configuration of particles. However, when particles are of finite volume, adjustments have to be made to the microscopic dynamics to avoid non-physical particle overlap. Under the assumption that particles are spherical and of radius $\varepsilon>0$, and interact through pairwise collisions, a binary collision rule can be introduced, which instantaneously changes the velocities of two particles so that they are moving away from each other when they come into contact: if at time $t$, particles $i$ and $j$ have positions $x_i,x_j$ and velocities $v_i^{\mathrm{in}},v_j^{\mathrm{in}}$ respectively, such that $(v_i^{\mathrm{in}}-v_j^{\mathrm{in}})\cdot(x_i-x_j)<0$ and $\|x_i-x_j\|=\varepsilon$, velocities are updated using:
	\begin{align}
	v_i^{\text{out}}=v_{i}^{\text{in}}-\frac{x_i-x_j}{\|x_i-x_j\|}\cdot(v_i^{\text{in}}-v_j^{\text{in}})\frac{x_i-x_j}{\|x_i-x_j\|},\nonumber\\
	v_j^{\text{out}}=v_{j}^{\text{in}}+\frac{x_i-x_j}{\|x_i-x_j\|}\cdot(v_i^{\text{in}}-v_j^{\text{in}})\frac{x_i-x_j}{\|x_i-x_j\|},\label{eq:CollisionRule}
	\end{align}
	{\em i.e.} the velocity components in the direction of $x_i-x_j$ are swapped and reflected. 
	
	At an informal level, the collisional effect is not recognised at the level of the Liouville equation, but is derived in the BBGKY hierarchy \cite{Balescu1975, Simonella2011}, for example as a consequence of additional assumptions on an interaction force \cite{Huang1987}. However, by instantaneously changing the velocities of two particles that undergo a collision, a fundamental assumption in the Liouville derivation is no longer valid; the dynamics of an individual particle are no longer smooth. We therefore cannot rely on the Liouville equation and must resort to an alternative formulation to derive an equation for $f^{(N)}(X,V,t)$. A suitable alternative is the {\em weak formulation} of the Liouville equation \cite{Wilkinson2018}.

	In \cite{Wilkinson2018}, a system of $N=2$ spherical particles with diameter $\varepsilon>0$ and $G(X,V,t)=0$ is fully characterised, and, under no additional assumptions, it is shown that
	for all smooth, compactly supported test functions $\Phi$, given initial data $f_0^{(2)}\in C^0(\mathcal{D})\cap L^1(\mathcal{D})$, such that $f_0^{(2)}$ integrates to 1 on the phase space, and is always positive, there exists a unique $f^{(2)}\in C^0((-\infty,\infty),L^1(\mathcal{D}))$ which satisfies
	\begin{gather}
	\int_{\mathcal{P}}
	\int_{\mathbb{R}^6}
	\int_{-\infty}^{\infty}
	\left[
	\frac{\partial\Phi(X,V,t)}{\partial t} + (V\cdot\nabla_X)\Phi(X,V,t)
	\right]f^{(2)}(X,V,t)
	\diff t \diff V \diff X
	\nonumber\\= \nonumber\\
	-
	\int_{\partial\mathcal{P}}
	\int_{\mathbb{R}^6}
	\int_{-\infty}^{\infty}
	\Phi(X,V,t)f^{(2)}(X,V,t)V\cdot\hat{\nu}(X)
	\diff t \diff V \diff \mathcal{H}(X)\label{eq:LiouvilleLinearIntegrals}
	\end{gather}
	and conserves linear and angular momentum, and kinetic energy. Here
	\begin{itemize}
		\item We define the spacial integral on $\mathcal{P} = \{X\in\mathbb{R}^6: \|x_1-x_2\|\ge\varepsilon\}$, as particles cannot overlap,
		\item $\mathcal{D}\in\mathbb{R}^{12}$ represents all possible configurations of positions and velocities, given that two particles cannot overlap ($X\in\mathcal{P}$).
		\item The vector $\hat{\nu}\in\mathbb{R}^6$ is the outward unit normal to the surface $\partial\mathcal{P}$,
		\item $\mathcal{H}(Y)$ is the Hausdorff measure \cite{Evans2015} on $\partial\mathcal{P}$.
	\end{itemize}
	In this case, we say that $f^{(2)}$ is a global-in-time weak solution of the Liouville equation:
	\begin{align}
	\frac{\partial f^{(2)}}{\partial t}  + (V\cdot\nabla_X)f^{(2)} = C[f^{(2)}]\label{eq:LiouvilleLinear}
	\end{align}
	where $C[f^{(2)}]$ is determined in the weak sense against test functions $\Phi\in C_1^1((-\infty,\infty),\mathcal{D})$:
	\begin{align}
	\langle C[f^{(2)}],\Phi\rangle = \int_{\partial\mathcal{P}}\int_{\mathbb{R}^6}f^{(2)}(X,V,t)\Phi(X,V,t)V\cdot\hat{\eta}_Y\diff V\diff \mathcal{H}(Y) \label{eq:LiouvilleLinearCollision}
	\end{align}
	In contrast to other formulations, a collisional term $C[F]$ is derived at the level of the Liouville equation. Furthermore, under the assumptions of molecular chaos, $C[F]$ admits the elastic Boltzmann collision operator \cite{Cercignani1988, Gallagher2012} in the first equation of the weak formulation of the BBGKY hierarchy, agreeing with previous results. Under restrictions on initial data, for example on the particle density of the system, an analogue to \cref{eq:LiouvilleLinear} should also hold in systems with more than 2 particles.
	
	From this point, a clear question to consider is how more complicated dynamics, or other instantaneous interactions between particles, affect the derivation. In this paper, we construct the weak formulation of the Liouville equation for two particles with general free dynamics ({\em i.e.} where there are no instantaneous interactions between particles), and general instantaneous interactions. The microscopic dynamics are first discussed in \Cref{Sec:MicroProperties}, which leads to the derivation of the Liouville equation in \Cref{Sec:LiouvilleFormulation}. This careful examination of the dynamics provides useful insights for mathematical modelling of materials approximated by hard particles, {\em e.g.} granular media. In particular the collision operators (in both position and velocity) constructed at the level of the Liouville equation in the weak formulation should be of interest.
	
	Following this, we consider an example in \Cref{subSec:InelasticExample} where collisions are {\em inelastic} and the free dynamics are affected by external friction and a constant external potential. Upon careful consideration of the weak Liouville equation, and the admissible initial data for the given dynamics, the example leads to a modified collisional term at the level of the BBGKY hierarchy. We discuss our findings and future directions for research in \Cref{Sec:Conclusions}.
	
	\subsection{Set-up}\label{subSec:SetUp}
	
	\subsubsection{Free particle dynamics}\label{subsubSec:FreeDynamics}
	At a microscopic level, we consider the initial data of two particles. The first particle has initial position $x\in\mathbb{R}^3$ and velocity $v\in\mathbb{R}^3$, the second has initial position $\bar{x}\in\mathbb{R}^3$ and velocity $\bar{v}\in\mathbb{R}^3$, all defined at time $t=0$. When considering the initial position of particles, we refer to the position of their centres of mass. It is useful to consider the concatenation of position and velocities as $X=[x,\bar{x}]\in\mathbb{R}^6$ and $V=[v,\bar{v}]\in\mathbb{R}^6$. In some cases it may be useful to define $Z=[X,V]\in\mathbb{R}^{12}$ in the same spirit.
	
	One of the important requirements of our method is an understanding of admissible initial data, {\em i.e.} initial data that produces a solution to \cref{eq:GeneralNewtons} for all time $t\in\mathbb{R}$. In this paper, we assume that the admissible initial position data for the free dynamics encompasses the entirety of $\mathbb{R}^6$, {\em i.e.} any initial positions can provide dynamics that are defined for all times $t\in\mathbb{R}$. However, the initial position data $X\in\mathbb{R}^6$ may restrict the admissible initial velocity data to a subset of $\mathbb{R}^6$. Therefore for each $X\in\mathbb{R}^6$, we define $\mathcal{V}^{\mathrm{f}}(X)\subseteq\mathbb{R}^6$ to be the set of initial velocity data which produces a solution to \cref{eq:GeneralNewtons} for all times $t\in\mathbb{R}$, and take $V\in\mathcal{V}^{\mathrm{f}}(X)$.
	
	Many calculations in this derivation refer to relative differences of position and velocity of the two particles, for example the binary collision rule \cref{eq:CollisionRule} is used in the dynamics when $(x-\bar{x})\cdot(v-\bar{v})<0$. Thus we introduce the following notation: for a vector $A=[a,\bar{a}]\in\mathbb{R}^6$, we write $\tilde{a} = a-\bar{a}\in\mathbb{R}^3$ as its relative difference. Much of the intuition in the derivation can be considered in terms of relative differences of particle data. For example, for the linear dynamics where $G(X,V,t)=0$ where particles are hard spheres with diameter $\varepsilon$, we can rewrite the dynamics in terms of relative differences, which effectively fixes one point at position $0$ with radius $\varepsilon$, and a collision is seen as a reflection of an intersecting point-particle trajectory.
	
	%
	
	Given initial data $X\in\mathbb{R}^6$ and $V\in\mathcal{V}^{\mathrm{f}}(X)$, we define the free particle flow maps for position and velocity by $\Phi_t^x(X,V)$ and $\Phi_t^v(X,V)$ respectively, and assume that they satisfy the Newton equations
	\begin{align}
	\partial_t\Phi_t^x(X,V,t) = \Phi_t^v(X,V,t),\quad \partial_t\Phi_t^v(X,V,t) = G(\Phi_t^x(X,V,t),\Phi_t^v(X,V,t),t).
	\end{align}
	We write $\Phi_t = [\Phi_t^x,\Phi_t^v]$ as the flow map for initial data $Z\in\mathcal{D}$. We also assume that the flow maps produce unique trajectories for any given initial data.
	
	\subsubsection{Dynamics with instantaneous interactions}\label{subsubSec:Dynamics}
	When instantaneous interactions are considered, a careful understanding of the admissible data is required. Outside of the discrete set of interaction times the particles follow trajectories determined by $\Phi_t^x$ and $\Phi_t^v$. In this section we will construct the flow maps $\Psi_t^x$ and $\Psi_t^v$ that include instantaneous interactions.
	
	\paragraph{Admissible data}
	In both of the examples we consider, the particles are hard spheres with diameter $\varepsilon$. The possible initial data for $X$ is therefore restricted to the {\em hard sphere table}
	\begin{align}
	\mathcal{P}_\varepsilon = \{X = [x,\bar{x}]\in\mathbb{R}^6 : \|x-\bar{x}\|\ge\varepsilon\}.
	\end{align}
	The subspace $\mathcal{P}_\varepsilon$ is a real analytic manifold with boundary $\partial\mathcal{P}_\varepsilon$, which has unit norm vector
	\begin{align}
	\hat{\nu} = \frac{1}{\sqrt{2}\varepsilon}\left[
	-\tilde{x},\tilde{x}
	\right]\label{eq:BoundaryNormal}
	\end{align}
	for $X\in\partial\mathcal{P}_\varepsilon$.
	
	We assume that the motion of the spherical particles is non-rotational, i.e. we do not furnish the equations of motion with an evolution equation for the angular velocity of the spheres. Adopting the assumption of smooth spherical particles is very popular in the literature, the introduction of particles which are non-spherical is also of interest and has been studied, for example computationally in \cite{Donev2005}.
	
	The introduction of instantaneous interactions between particles will change what initial velocities are admissible. We define the set of admissible velocity data for dynamics with instantaneous interactions as $\mathcal{V}(X)$ for each $X\in\mathcal{P}$, which we assume has a piecewise analytic boundary. Furthermore, we define $\mathcal{C}(X)\subset\mathcal{V}(X)$ to be the set of initial data which leads to instantaneous interactions, and also assume that $\partial\mathcal{C}(X)$ is a piecewise analytic submanifold of $\mathbb{R}^6$. We validate this in the two examples considered.
	
	We may consider additional interaction diameters $\mathcal{P}_{\tilde{\varepsilon}}$ for $\tilde{\varepsilon}>\varepsilon$, {\em e.g.} square well interactions discussed in \cite{Bannerman2010}, but importantly for $X\in\mathbb{R}^6\backslash\mathcal{P}_\varepsilon$, $\mathcal{V}(X)=\emptyset$.
	
	We assume that the dynamics with instantaneous interactions have flow maps $\Psi_t^x(X,V)$ and $\Psi_t^v(X,V)$, defined globally in time. 
	
	\paragraph{Event times}
	We will characterise each instantaneous event by an event time, an interaction diameter and an event map that changes the particle velocities. As interactions are instantaneous, the event times can be enumerated as a discrete set, which can be finite or infinite. We write the event times as $\tau_i\in\mathbb{R}$ for $i=-M,-M+1,...,-1,0,1,...,N-1,N$, where $M=M(X,V),N=N(X,V)\in\mathbb{N}\cup\{\infty\}$, and
	\begin{align}
	-\infty=\tau_{-M}(X,V)<\tau_{-M+1}(X,V)<....<\tau_{N-1}(X,V)<\tau_N(X,V)=\infty,
	\end{align}
	where we choose $\tau_0$ to be the closest event time to time $t=0$. These events occur when particles reach an interaction diameter. As particles may have many interaction diameters, to each event time $\tau_i$ we associate a particular interaction diameter $\varepsilon_i$ for each $i=1,...,N-1$, and let
	\begin{align}
	\partial\mathcal{P}_i = \{X\in\mathcal{P}:\|\tilde{x}\|=\varepsilon_i\}.
	\end{align}
	Then each $\tau_i$ can be defined as the first time the particles reach an interaction diameter after the previous event.
	\begin{align}
	\tau_i(X,V) = \{s\in(\tau_{i-1},\tau_{i+1}):\Phi_{s-\tau_{i-1}}(\Psi_{\tau_{i-1}}^x,\Psi_{\tau_{i-1}}^v)\in\partial\mathcal{P}_i\}.
	\end{align}
	We will assume that there exists $\delta>0$ such that for all $i=1,...,N$, $\tau_i-\tau_{i-1}>\delta$, and we have defined two special event times $\tau_0(X,V)=-\infty$, $\tau_N(X,V)=\infty$ as the behaviour as $t\rightarrow\infty$. Between each pair of event times, the particle dynamics are determined by \cref{eq:GeneralNewtons}. At each time $\tau_i$, the particles experience an instantaneous change in velocity. For hard spheres, for example, this will ensure that the two particles do not overlap.
	
	\paragraph{Event maps}
	At each time $\tau_i$, the velocities of the particles experience an instantaneous change. It is necessary for the change in velocities to conserve linear and angular momentum, {\em i.e.} if an event $\tau_i(X,V), X\in\mathcal{P},V\in\mathcal{V}(X)$ occurs at time $t=0$ (without loss of generality), then
	\begin{align}
	v' + \bar{v}' = v + \bar{v}, \label{eq:COLM}
	\end{align}
	and for all $a\in\mathbb{R}^3$,
	\begin{align}
	(x-a)\times v' + (\bar{x}-a)\times \bar{v}' = (x-a)\times v + (\bar{x}-a)\times \bar{v}, \label{eq:COAM}
	\end{align}
	where primed velocities denote post event velocities. In fact, to show conservation of angular momentum one only needs to check \cref{eq:COAM} is satisfied for 4 values of $a$.
	
	\begin{proposition}\label{prop:Polytope}
		\cref{eq:COAM} and \cref{eq:COLM} are true for all $a\in\mathbb{R}^3$ if and only if 
		\begin{align}
		(x-p_i)\times v' + (\bar{x}-p_i)\times \bar{v}' = (x-p_i)\times v + (\bar{x}-p_i)\times \bar{v}, \label{eq:COAMPolytope}
		\end{align}
		for $i=1,...,4$, where $\{p_i\}_{i=1}^4\subset\mathbb{R}^3$ are the vertices of a (non-degenerate) polytope in $\mathbb{R}^3$.
	\end{proposition}
	\begin{proof}
		The necessity is a trivial result. For sufficiency, by \cref{eq:COLM} we may assume $x=0$ without loss of generality, and so for each $p_j$, for all constants $c_j\in\mathbb{R}$
		\begin{align*}
		&-c_jp_j\times v' + c_j(\bar{x}-p_j)\times\bar{v}'
		=
		-c_jp_j\times v + c_j(\bar{x}-p_j)\times\bar{v},\\
		\implies&
		-\sum_{j=1}^4c_jp_j\times v' + \sum_{j=1}^4c_j(\bar{x}-p_j)\times\bar{v}'
		=
		-\sum_{j=1}^4c_jp_j\times v +\sum_{j=1}^4 c_j(\bar{x}-p_j)\times\bar{v}.
		\end{align*}
		If we now suppose that $\sum_{j=1}^4c_j=1$, then we have that \cref{eq:COAM} is satisfied for all $q$ in the set
		\begin{align}
		\mathcal{C} = \left\{\sum_{i=1}^4c_jp_j: \sum_{j=1}^4c_j=1\right\}
		\end{align}
		As $\mathcal{C}$ is a convex set in $\mathbb{R}^3$, we infer that $\mathcal{C}=\mathbb{R}^3$, as required.
	\end{proof}
	Considering $\mathbb{R}^6$ notation, the change in velocity is determined by a map $\sigma_i(X,V):\mathbb{R}^6\rightarrow\mathbb{R}^6$, where $\sigma_i(X,V)\in\mathbb{R}^6\times\mathbb{R}^6$. We have the following result on the form of $\sigma(X,V)$.
	
	\begin{Theorem}\label{Theorem:sigma}
		Let $X\in\mathbb{R}^6$ and $V\in\mathbb{R}^6$, and set $X=[x,\bar{x}],V=[v,\bar{v}]$. Assume that	$\bar{x}\neq x$, and set
		\begin{align}
		N(X) = \frac{1}{\|x-\bar{x}\|}[x-\bar{x},\bar{x}-x].
		\end{align}
		Then the following are equivalent:
		\begin{enumerate}
			\item 			
			\begin{align}
			\sigma(X,V) = I - \eta(X,V)N(X)\otimes N(X),\label{sigma}
			\end{align}
			for some $\eta(X,V):\mathbb{R}^{12}\rightarrow\mathbb{R}$.
			
			\item		
			The map $\sigma(X,V)$ satisfies linear and angular momentum:
			\begin{enumerate}
				\item (COLM)
				\begin{align}
				v' + \bar{v}' = v + \bar{v},\label{COLM}
				\end{align}
				\item (COAM)
				For any $a\in\mathbb{R}^3$
				\begin{align}
				(x-a)\times v' + (\bar{x}-a)\times \bar{v}'
				=
				(x-a)\times v + (\bar{x}-a)\times \bar{v}\label{COAM}
				\end{align}
				where $v'=(\sigma(X,V)V)_{(1,2,3)}$ and similar for $\bar{v}'$.
			\end{enumerate}
		\end{enumerate}
	\end{Theorem}
	\begin{proof}
		We note that \cref{COAM} can be written as $A\sigma(X,V)V=AV$, where
		\begin{align}
		A_a=\begin{pmatrix}
		0      & -x_3^a & x_2^a  & 0            & -\bar{x}_3^a & \bar{x}_2    \\
		x_3^a  & 0      & -x_1^a & \bar{x}_3^a  & 0            & -\bar{x}_1^a \\
		-x_2^a & x_1^a  & 0      & -\bar{x}_2^a & \bar{x}_1^a  & 0      
		\end{pmatrix},
		\end{align}
		where we have written $x_i^a = x_i-a$.
		\begin{itemize}
			\item (1 $\implies$ 2). This can be shown by a direct calculation. In short, we see that
			\begin{align*}
			A+a\sigma(X,V)&=A_a(I - \eta(X,V)N\otimes N),\\
			&= A_a - \eta(X,V)A_aN(X)\otimes N(X),\\
			&= A_a,
			\end{align*}
			where we have found that $A_aN(X)=(0,0,0)$.
			\item (2 $\implies$ 1). Before starting the calculation, we note that the matrix $A_a$ is rank 3 (by considering its row-echelon form), and that it is enough to show that \cref{COAM} holds for $a=\{a_1,a_2,a_3,a_4\}$, if \cref{COLM} holds.		
			
			We note that $A_a\sigma(X,V)V=A_aV$ for all $V\in\mathbb{R}^6$ implies that $A_a(\sigma(X,V)V-V)=0$. Therefore $\sigma(X,V)V-V\in\mathrm{ker}(A_a)$ for all $a\in\mathbb{R}^6$. Let $Y=(\sigma(X,V)V-V)$, then
			\begin{gather*}
			- Y_2x_3^a + Y_3x_2^a - Y_5\bar{x}_3^a + Y_6\bar{x}_2^a = 0,\\
			Y_1x_3^a - Y_3x_1^a + Y_4\bar{x}_3^a - Y_6\bar{x}_1^a = 0,\\
			- Y_1x_2^a + Y_2x_1^a - Y_4\bar{x}_2^a + Y_5\bar{x}_1^a = 0.
			\end{gather*}
			We consider the values 
			\begin{gather*}
			a_1 = (x_1,x_2,x_3),\\
			a_2 = (x_1,x_2,\bar{x}_3),\\
			a_3 = (x_1,\bar{x}_2,\bar{x}_3),\\
			a_4 = (\bar{x}_1,\bar{x}_2,\bar{x}_3),
			\end{gather*}
			which form the four vertices of a tetrahedron, and result in
			\begin{align}
			Y = \tilde{\eta}(x,v)\begin{pmatrix}
			x_1 - \bar{x}_1\\
			x_2 - \bar{x}_2\\
			x_3 - \bar{x}_3\\
			\bar{x}_1 - x_1\\
			\bar{x}_2 - x_2\\
			\bar{x}_3 - x_3
			\end{pmatrix} = \eta(X,V) N(X).\label{YN(X)},
			\end{align}
			where $\eta(X,V) = \|x-\bar{x}\|\tilde{\eta}(X,V)$, without loss of generality. We note, taking the dot product on both sides of \cref{YN(X)} with $N(X)$ and rearranging, that
			\begin{align*}
			&N(X)\cdot(\sigma(X,V)-V) = \eta(X),\\
			\implies&
			(I-N(X)\otimes N(X))\sigma(X,V)V = (I - N(X)\otimes N(X)) V,
			\end{align*} 
			{\em i.e.} $\sigma(X,V)$ can only change the component of $V$ in the direction of $N(X)$. Thus, if $N(X)\cdot V=0$, then $V$ is contained in the hyperplane orthogonal to $N(X)$, and we must have that $\sigma(X,V)V=V$. Therefore, without loss of generality, for any $V\in\mathbb{R}^6$ we can take $\eta(X,V) = N(X)\cdot V\tilde{\tilde{\eta}}(X,V)$, and so
			\begin{align}
			\sigma(X,V)V = (I-\eta(X,V)N(X)\otimes N(X))V,
			\end{align}
			which is as required.
		\end{itemize}
	\end{proof}
	
	\Cref{eq:COLM,eq:COAM} are not enough to fully determine the map $\sigma(X,V)$; an additional constraint must be supplied. For example, in \cite{Wilkinson2018}, it is shown that the Boltzmann (elastic) scattering map can be determined by including the conservation of kinetic energy. In this case $\sigma(X,V)$ is an involution, that is $\sigma(X,V)^2 = I$. Alternative maps can be derived using different constraints on the Jacobian of the scattering map (for inelastic Boltzmann scattering maps considered in \Cref{subSec:InelasticExample}), or on kinetic energy (for boost or damping maps considered in \cite{Bannerman2010}). We call this the {\em event map constraint}.
	
	We define the forward time map $\sigma^+_i(X,V)$ as the map which takes pre-event velocities to post-event velocities, and the backward time map $\sigma^-_i(X,V)$ taking post-event velocities to pre-event velocities. For the dynamics to be reversible we require $\sigma_i^+(X,V)\sigma_i^-(X,V)=I$. Note that if we assume that $\sigma_i^-(X,V)=\sigma_i^+(X,V)$ then $\eta(X,V)=0$ or $\eta(X,V)=-2$. The former value of $\eta(X,V)$ produces the identity map, while the latter is the elastic Boltzmann scattering map.
	
	We can fully define the flow maps for dynamics with instantaneous interactions, using $\Phi_t^x,\Phi_t^v,\tau_i$ and $\sigma_i^\pm$. We split the initial data into three cases.
	
	\paragraph{No instantaneous events.} If $X,V$ are such that no instantaneous events happen, then $M+N=1$ and $\Psi_t^X,\Psi_t^v$ obey
	\begin{align}
	\partial_t\Psi_t^x(X,V) = \Phi_t^x(X,V),\\
	\partial_t\Psi_t^v(X,V) = G(X,V,t).
	\end{align}
	\paragraph{Instantaneous events.} When $X,V$ are such that instantaneous events occur in finite time, then
	\begin{align}
	\partial_t\Psi_t^x(X,V) = \begin{cases}
	\Phi_t^v(X,V), & \tau_{i_0-1}\le t \le \tau_{i_0},\\
	\Phi_{t-\tau_i}^v(\Psi_{\tau_i}^x(X,V),\sigma_i^+(X,V)\Psi_{\tau_i}^v(X,V), & 
	\begin{aligned}[c]&\tau_i<t\le\tau_{i+1},\\ &i=i_0,\dots,N-1,
	\end{aligned}
	\\
	\Phi_{t-\tau_i}^v(\Psi_{\tau_i}^x(X,V),\sigma_i^-(X,V)\Psi_{\tau_i}^v(X,V), & \begin{aligned}[c]
	&\tau_{i-1}\le t<\tau_{i},\\
	&i=-(M+1),\dots,(i_0-1).
	\end{aligned}
	\\
	\end{cases}
	\end{align}
	and
	\begin{align}
	\partial_t\Psi_t^v(X,V) = \begin{cases}
	G(X,V,t), & \tau_{i_0-1}\le t \le \tau_{i_0},\\
	G(\Psi_{\tau_i}^x(X,V),\sigma_i^+(X,V)\Psi_{\tau_i}^v(X,V),t-\tau_i), & \begin{aligned}[c]&\tau_i<t\le\tau_{i+1},\\& i=i_0,\dots,N-1,
	\end{aligned}
	\\
	G(\Psi_{\tau_i}^x(X,V),\sigma_i^-(X,V)\Psi_{\tau_i}^v(X,V),t-\tau_i), & \begin{aligned}[c]&\tau_{i-1}\le t<\tau_{i},\\& i=-(M+1),\dots,(i_0-1).
	\end{aligned}
	\\
	\end{cases}
	\end{align}
	where $i_0=1$ if $\tau_0<0$ and $\mu=0$ if $\tau_0>0$.
	
	Note that free particle flow maps need not be defined globally to be used in these flow maps, if the instantaneous interaction renders the trajectory admissible.
	
	We now have all the necessary notation to state the main results of this paper.
	\subsection{Main results}\label{subSec:MainResults}
	Now that we have fully defined the dynamics of two particles, we are in a position to state the main results of this paper. To do so, we state the following definition, which is generalised from \cite{Wilkinson2018}.
	
	\begin{Definition}(Global in time weak solutions of the Liouville equation) Suppose we are given an initial condition $f_0\in C^0(\mathcal{D})\cap L^1(\mathcal{D})$, such that
		\begin{align}
		\int_{\mathcal{P}}\int_{\mathcal{V}(X)}f_0(X,V)\diff V\diff X = 1,\quad f_0(X,V)\ge0.
		\end{align}
		Then $f\in C^0((-\infty,\infty),L^1(\mathcal{D}))$ is a physical global in time solution of the Liouville equation
		\begin{align}
		\partial_tf + V\cdot\nabla_Xf +\nabla_V\cdot(G(X,V,t)f) = C_X[f^{(2)}] + C_V[f^{(2)}]
		\end{align}
		if and only if for all test functions $\Phi\in C_c^1(T\mathbb{R}^6\times (-\infty, \infty))$
		\begin{align}
		&\int_{\mathcal{P}}\int_{\mathcal{V}(X)}\int_{-\infty}^{\infty}
		f(X,V,t)
		[\partial_t\Phi(X,V,t) + V\cdot\nabla_X\Phi(X,V,t) \nonumber\\& 
		\qquad \qquad \qquad \qquad \qquad \qquad
		+ \nabla_V\cdot(G(X,V,t)\Phi(X,V,t))]
		\diff t\diff V\diff X \nonumber\\ &
		=\nonumber\\
		&-\int_{\partial\mathcal{P}}\int_{\mathcal{V}(X)}\int_{-\infty}^{\infty}
		f(X,V,t)\Phi(X,V,t)V\cdot\hat{\nu}_X
		\diff t\diff V\diff \mathcal{H}(X)\nonumber\\
		&-\int_{\mathcal{P}}\int_{\partial\mathcal{V}(X)}\int_{-\infty}^{\infty}
		f(X,V,t)\Phi(X,V,t)G(X,V,t)\cdot\hat{\nu}_V
		\diff t\diff \mathcal{H}(X,V)\diff X,\label{eq:LiouvilleGeneral}
		\end{align}
		and $f$ obeys the conservation of linear and angular momentum for all $t\in(-\infty,\infty)$:
		\begin{align}
		&\int_{\mathcal{P}}\int_{\mathcal{V}(X)}(v+\bar{v})f(X,V,t)\diff V\diff x
		=
		\int_{\mathcal{P}}\int_{\mathcal{V}(X)}(v+\bar{v})f_0(X,V)\diff V\diff x,\\\
		&\int_{\mathcal{P}}\int_{\mathcal{V}(X)}(x\times v+\bar{x}\times\bar{v})f(X,V,t)\diff V\diff x
		=
		\int_{\mathcal{P}}\int_{\mathcal{V}(X)}(x\times v+\bar{x}\times\bar{v})f_0(X,V)\diff V\diff x,
		\end{align}
		and the microscopic dynamics satisfy the associated event map constraints. 
	\end{Definition}
	With this, we state the main result of this paper.
	\begin{Theorem}[Existence of weak global in time solutions to the Liouville equation]\label{Theorem:main}
		For any $f_0\in C^0(\mathcal{D})\cap L^1(\mathcal{D})$, there exists a physical global in time weak solution to the Liouville equation \cref{eq:LiouvilleGeneral}.
	\end{Theorem}
	Before identifying important transport identities and proving \Cref{Theorem:main}, we make some remarks on this result.
	
	Firstly, we have made very few assumptions on the free dynamics of the particles; they can be affected by external or interparticle forces. In \Cref{subSec:InelasticExample} we consider inelastic collisions, but the same formulation can be used to consider any interactions determined by discrete step potentials \cite{Bannerman2010}. The result is therefore quite general, and should be appropriate for a range of systems.
	
	The Liouville equation is derived for a system of two particles. In principle, systems of many particles may involve many body interactions. However, given the correct subset of initial data, which ensures that all instantaneous interactions are pairwise, the equation \cref{eq:LiouvilleGeneral} should also be accurate for systems of many particles, and may be used to approximate systems where the initial data is not so carefully constructed.
	
	We can also state a general form of the BBGKY hierarchy. We start with a definition of global in time weak solutions.
	
	\begin{Definition}\label{def:BBGKY}
		Let $f_0\in C^1(\mathcal{D})\cap L^1(\mathcal{D})$ be symmetric in its particle arguments ({\em i.e.} $[v,x]$ and $[\bar{v},\bar{x}]$ can be interchanged). We say that a pair of maps $(f^{(1)}_0,f^{(2)}_0)$ with membership
		\begin{align}
		f^{(1)}\in C^0((-\infty,\infty),L^1(T\mathbb{R}^3)),\quad
		f^{(2)}\in C^0((-\infty,\infty),L^2(\mathcal{D}))
		\end{align}
		is a global in time weak solution of the BBGKY hierarchy associated to the initial data
		\begin{align}
		f_0^{(1)} = \int_{\mathbb{R}^3\backslash\mathcal{B}_\varepsilon(x)}\int_{\mathcal{V}(x,\bar{x},v)}
		f_0(X,V)\diff\bar{v}\diff\bar{x},\quad \text{ for all } [x,v]\in T\mathbb{R}^3,
		\end{align}
		and
		\begin{align}
		f_0^{(2)}(X,V) = f_0(X,V),\quad\text{ for all }[X,V]\in\mathcal{D},
		\end{align}
		if and only if, for all test functions $\phi\in C_c^\infty(T\mathbb{R}^3\times(-\infty,\infty))$,
		\begin{gather}
		\int_{\mathbb{R}^3}\int_{\mathbb{R}^3}\int_{-\infty}^{\infty}(\partial_t+v\cdot\nabla_x)\phi(x,v,t)f^{(1)}(x,v,t)\diff t\diff v\diff x\nonumber\\
		+
		\int_{\mathbb{R}^3}\int_{\mathbb{R}^3\backslash\mathcal{B}_\varepsilon(x)}\int_{\mathbb{R}^3}\int_{\mathcal{V}(x,\bar{x},v)}\int_{-\infty}^\infty\nabla_V\cdot(G(X,V,t)\phi(x,v,t))
		\nonumber\\
		\qquad \qquad \qquad \qquad \qquad \qquad \qquad 
		\times f^{(2)}(X,V,t)\diff t\diff\bar{v}\diff v\diff\bar{x}\diff x\nonumber\\
		=\label{eq:BBGKYGeneral}\\
		-\frac{1}{\sqrt{2}}
		\int_{\mathbb{R}^3}\int_{\mathbb{S}^2}\int_{\mathcal{V}(X)}\int_{-\infty}^{\infty}
		\phi(x,v,t)f^{(2)}([x,\bar{x}+\varepsilon n],[v,\bar{v}],t)(v-\bar{v})\cdot n \diff t\diff V\diff n \diff x\nonumber\\
		-\int_{\mathcal{P}}\int_{\partial\mathcal{V}(X)}\int_{-\infty}^{\infty}
		\phi(x,v,t)f^{(2)}(X,V,t)G(X,V,t)\cdot\hat{\nu}_V\diff t\diff\mathcal{H}(V,X)\diff X,
		\end{gather}
		and $f^{(2)}$ satisfies \cref{eq:LiouvilleGeneral}.
	\end{Definition}
	After a straightforward partition of the phase space, the BBGKY hierarchy then follows as a corollary of \Cref{Theorem:main}. For more applicable results, one must have a good understanding of the admissible data. For example, in linear elastic hard sphere dynamics this would result in the Boltzmann collision operator on the right hand side of \cref{eq:BBGKYGeneral} \cite{Wilkinson2018}. The following is a quick corollary of our main result.
	\begin{Corollary}\label{Cor:BBGKY}
		For any $f_0\in C^0(\mathcal{D})\cap L^1(\mathcal{D})$ there exists a global in time solution to the BBGKY hierarchy given by \Cref{def:BBGKY}.
	\end{Corollary}
	
	Before formulating the Liouville equations we state and prove some microscopic properties that are used in the derivation.
	\section{Microscopic properties}\label{Sec:MicroProperties}
	The proof of \Cref{Theorem:main} relies on transport identities for $\tau_i(X,V)$ and $\sigma_i(X,V)$. These identities are given in the case for linear elastic particles in \cite{Wilkinson2018}. Here we generalise to our case and interpret the results physically.
	
	We first consider the results for the free dynamics.
	
	\begin{proposition}
		[Transport Identity I]\label{prop:TIdI}
		Given $X\in\mathbb{R}^6$ and $V\in\mathcal{V}(X)$, let $\tau(X,V)$ be a particular event time in the dynamics determined by $\Psi_t^x(X,V)$ and $\Psi_t^v(X,V)$. Then for all $t\in(\tau_{i-1},\tau_{i+1})$, (omitting flow map arguments for ease of notation)
		\begin{align}
		[\Phi_t^v\cdot\nabla_{Y} + G(\Phi_t^x,\Phi_t^v)\cdot\nabla_W]\tau(Y,W)|_{Y=\Psi_t^v,W=\Psi_t^v}=-1.\label{eq:TI1}
		\end{align}
	\end{proposition}
	\begin{proof}
		For ease of notation, we omit the arguments of $\Phi_t^x$ and $\Phi_t^v$. Firstly, we note that the event time $\tau(X,V)$ can be written as \cite{Wilkinson2018}
		\begin{align*}
		\tau(X,V) = \arg\min\{s:\Phi_t^x(X,V)\in\partial\mathcal{P}_i\}.
		\end{align*}
		We first consider $\tau$ as a function of the data at time $t$, i.e.
		\begin{align*}
		\tilde{\tau}(\Psi_t(X,V),\Psi_t(X,V)) = \tau(X,V).
		\end{align*}
		To construct the time derivative of $\tilde{\tau}$, we use first principles. Let $h>0$ and assume that $h\ll\delta$. Then
		\begin{align*}
		\tau(\Psi_{t+h}^x,\Psi_{t+h}^v)=&
		\arg\min\{
		s\in(\tau_{i-1},\tau_{i+1}):\|\tilde{\Phi}_s^x(\Psi_{t+h}^x,\Psi_{t+h}^v)\|=\varepsilon
		\}\\
		=&\arg\min\{
		s\in(\tau_{i-1},\tau_{i+1}):\|\tilde{\Phi}_{s-h}^x(\Psi_{t}^x,\Psi_{t}^v)\|=\varepsilon
		\}\\
		=&\arg\min\{
		s+h\in(\tau_{i-1},\tau_{i+1}):\|\tilde{\Phi}_{s}^x(\Psi_{t}^x,\Psi_{t}^v)\|=\varepsilon
		\}.
		\end{align*}
		We have assumed that there exists a unique value $\bar{s}$ for which this is true. It follows that
		\begin{align*}
		\frac{1}{h}(\tau(\Psi_{t}^x,\Psi_{t}^v)-\tau(\Psi_{t+h}^x,\Psi_{t+h}^v))
		=\frac{1}{h}(\bar{s}-(\bar{s}+h)) = -1.
		\end{align*}
		Using results from generator theory \cite{Pavliotis2008}, we find that, for $t\in(\tau_{i-1},\tau_{i+1})$:
		\begin{align*}
		\partial_t\tau(\Psi_t^x,\Psi_t^v)
		=
		\left[\partial_t\Phi_t^x\cdot\nabla_Y + \partial_t\Phi_t^v\cdot\nabla_W\right]\tau(Y,W)|_{Y=\Phi_t^x,w=\Phi_t^v}.
		\end{align*}
		By using the Newton equations for the free particle dynamics we arrive at the required result.
		
	\end{proof}
	The time derivative of $\tau$ offers an insight to the geometric meaning of the result; advancing forward in time towards an event in the future decreases the time until the event proportionally. The second transport identity of interest involves the event maps $\sigma_{\pm}$. We note that for linear dynamics \cref{eq:TI1} reduces to
	\begin{align}
	V\cdot\nabla_X\tau(X,V) = -1,
	\end{align}
	which is the first transport identity given in \cite{Wilkinson2018}. We can also provide an analogous result to the second transport identity in \cite{Wilkinson2018}.
	\begin{proposition}
		[Transport Identity II]\label{prop:TIdII}
		Given $X\in\mathbb{R}^6$ and $V\in\mathcal{V}(X)$, let $\sigma_i(X,V)$ be a particular event map in the dynamics determined by $\Psi_t^x(X,V)$ and $\Psi_t^v(X,V)$. Then for all $t\in(\tau_{i-1},\tau_{i+1})$, (omitting flow map arguments for ease of notation)
		\begin{align}
		[\Phi_t^v\cdot\nabla_{Y} + G(\Phi_t^x,\Phi_t^v)\cdot\nabla_W]\sigma(Y,W)|_{Y=\Psi_t^v,W=\Psi_t^v}=0.\label{eq:TI2}
		\end{align}
	\end{proposition}
	\begin{proof}
		Again, generator theory provides a link between the time derivative of $\sigma_i(X,V)$ and the left hand side of \cref{eq:TI2}. We note that, by shifting in time,
		\begin{align*}
		\partial_t\sigma(\Phi_t^x,\Phi_t^v) =&
		\partial_t\left(\eta[\Phi_{\tau_i-t}^x(\Psi_t^x,\Psi_t^v),\Phi_{\tau_i-t}^v(\Psi_t^x,\Psi_t^v)]
		\left(N[\Phi_{\tau_i-t}^x(\Psi_t^x,\Psi_t^v)]
		\otimes 
		N[\Phi_{\tau_i-t}^x(\Psi_t^x,\Psi_t^v)]\right)\right),\\
		=& \partial_t\left(\eta\left[\Phi_{\tau_i}^x,\Phi_{\tau_i}^v\right]
		N[\Phi_{\tau_i}^x]
		\otimes
		N[\Phi_{\tau_i}^x]\right)=0,
		\end{align*}
		which completes the proof.
	\end{proof}
	The right hand side confirms that the scattering maps are not dependent on time; they depend only on the instantaneous positions and velocities at the time of the interaction.
	
	We can now use these microscopic properties to construct the Liouville equation in the next section.
	\section{Liouville formulation}\label{Sec:LiouvilleFormulation}
	We consider the following three integrals
	\begin{align}
	I(\Phi) =& \int_{\mathcal{P}}\int_{\mathcal{V}(X)}\int_{-\infty}^\infty
	f^{(2)}(X,V,t)\partial_t\Phi(X,V,t)\diff t\diff V \diff X,\label{eq:timeLiouville}\\
	J(\Phi) =& \int_{\mathcal{P}}\int_{\mathcal{V}(X)}\int_{-\infty}^\infty
	f^{(2)}(X,V,t)V\cdot\nabla_X\Phi(X,V,t)\diff t\diff V \diff X,\label{eq:spaceLiouville}\\
	K(\Phi) =& \int_{\mathcal{P}}\int_{\mathcal{V}(X)}\int_{-\infty}^\infty
	f^{(2)}(X,V,t)\nabla_V\cdot[G(X,V,t)\Phi(X,V,t)]\diff t\diff V \diff X,\label{eq:velocityLiouville}
	\end{align}
	which we call the time, space and velocity derivative terms respectively. The method to derive the Liouville equation in both of our examples is similar to the method used in \cite{Wilkinson2018}: we wish to find weak solutions $f^{(2)}$ to the equation
	\begin{align*}
	\mathcal{L}[f^{(2)}] = C[f^{(2)}],
	\end{align*}
	where the operator $\mathcal{L}$ is the Liouville equation associated to our choice of dynamics, for point-like particles. To derive the operator on the right hand side, we consider the Liouville equation on test functions $\Phi\in C_c^\infty(\mathcal{D}\times(-\infty,\infty))$, multiply by $f^{(2)}$ and integrate to find
	\begin{align}
	\int_{\mathcal{P}}\int_{\mathbb{R}^6}\int_{-\infty}^{\infty}
	f^{(2)}(X,V,t)\mathcal{L}[\Phi(X,V,t)]\diff X \diff V \diff t = 0.
	\end{align}
	We then separate the phase space into parts where $f^{(2)}$ is smooth, and evaluate each of the integrals constructed. Surface terms that arise then contribute to the collisional term on the right hand side of the weak Liouville equation.

	\subsection{The time derivative term}\label{subSec:TimeDerivative}
	Using the flow maps $\Psi_t^x,\Psi_t^v$, we can write the integral $I(\Phi)$ as follows:
	\begin{align}
	I(\Phi) = \int_{\mathcal{P}}\int_{\mathcal{V}(X)}
	\sum_{i=-M(X,V)}^{N(X,V)}
	\int_{-\tau_i}^{-\tau_{i-1}}
	f_0^{(2)}(\Psi_{-t}^x,\Psi_{-t}^v)\partial_t\Phi(X,V,t)\diff t \diff V\diff x.
	\end{align}
	On each interval $(\tau_{i-1},\tau_i)$ the flow is described by the free dynamics, so we may apply integration by parts, keeping in mind that evaluating $\Psi_t(X,V)$ at event times from the left or right provides a different result, and using a compact function $\Phi$ yields zero at $\tau_{-M}=-\infty,\tau_N=\infty$,
	\begin{align*}
	I(\Phi) =& \int_{\mathcal{P}}\int_{\mathcal{V}(X)}
	\sum_{i=-M+1}^{N-1}
	\Bigg\{
	\Phi(X,V,\tau_{i})(f_0^{(2)}(\Psi_{\tau_{i}^-}) - f_0^{(2)}(\Psi_{\tau_{i}^+}))\\
	&-
	\int_{-\tau_i}^{-\tau_{i-1}}
	\Phi(X,V,t)\partial_tf_0^{(2)}(\Psi_{-t}^x,\Psi_{-t}^v)\diff t
	\Bigg\}
	\diff V \diff X,
	\end{align*}
	where
	\begin{align*}
	\Psi_{\tau^-} = \lim_{t\rightarrow\tau^-}\Psi_t,\quad
	\Psi_{\tau^+} = \lim_{t\rightarrow\tau^+}\Psi_t.
	\end{align*}
	meaning that the summation of the surface terms in this integral do not cancel (due the the application of $\sigma_i(X,V)$ at each event time $\tau_i$). For each interval $(\tau_{i-1},\tau_i)$, the flow maps $\Psi_t^x,\Psi_t^v$ are determined by the free dynamics $\Phi_t^x,\Psi_t^v$ with particular initial conditions. We apply the chain rule on each of these intervals. For the second term in this result, we consider three separate cases.
	
	If there are no particle-particle interactions for initial data $X,V$ such that $M+N=1$, we have that
	\begin{align*}
	\partial_tf_0^{(2)}(\Psi_{-t}^x,\Psi_{-t}^v)
	=&
	[\partial_t\Phi_{-t}^x\cdot\nabla_Y + \partial_t\Phi_{-t}^v\cdot\nabla_W]
	f_0^{(2)}|_{Y=\Phi_{-t}^x,W=\Phi_{-t}^v},\\
	=&
	[\Phi_{-t}^v\cdot\nabla_Y + G(X,V,-t)\cdot\nabla_W]
	f_0^{(2)}|_{Y=\Phi_{-t}^x,W=\Phi_{-t}^v},\\
	=&-[\Psi_{-t}^v\cdot\nabla_X + G(X,V,-t)\cdot\nabla_V]
	f^{(2)}(X,V,t).
	\end{align*}
	
	For initial data which causes particle-particle collisions,
	\begin{align*}
	\partial_tf_0^{(2)}(\Psi_{-t}^x,\Psi_{-t}^v)
	=&
	\begin{cases}
	-[\Psi_{-t}^v\cdot\nabla_X + G(X,V,-t)\cdot\nabla_V]f^{(2)}(X,V,t),&\tau_{i_0-1}\le t\le\tau_{i_0},\\
	-[\Psi_{-(t-\tau_i)}^v\cdot\nabla_Y + G(\Psi_{-\tau_i}^x,\sigma_i^+\Psi_{-\tau_i}^v,-(t-\tau_i))\cdot\nabla_W]f^{(2)}(Y,W,t),&\begin{aligned}[c]&\tau_i<t<\tau_{i+1},\\& i=i_0,...,N-1,\end{aligned}\\
	-[\Psi_{-(t-\tau_i)}^v\cdot\nabla_Y + G(\Psi_{-\tau_i}^x,\sigma_i^-\Psi_{-\tau_i}^v,-(t-\tau_i))\cdot\nabla_W]f^{(2)}(Y,W,t),&\begin{aligned}[c]&\tau_i<t<\tau_{i+1},\\& i=-M+1,...,(i_0-1),\end{aligned}\\
	\end{cases}
	\end{align*}
	where $i_0=0$ if $\tau_0<0$ ({\em i.e.} the closest event to $t=0$ is in the past) and $i_0=1$ if $\tau_0>0$ (the closes event is in the future).
	
	We partition the result into the surface terms and the new integrals:
	
	\begin{gather}
	I_1(\Phi) = \int_{\mathcal{P}}\int_{\mathcal{V}(X)}\sum_{i=-M+1}^{N-1}
	\Phi(X,V,\tau_{i})(f_0^{(2)}(\Psi_{\tau_{i}^-}) - f_0^{(2)}(\Psi_{\tau_{i}^+}))\diff V \diff X,\\
	I_2(\Phi) = -\int_{\mathcal{P}}\int_{\mathcal{V}(X)}\sum_{i=-M+1}^{N-1}
	\int_{-\tau_i}^{-\tau_{i-1}}
	\Phi(X,V,t)\partial_tf_0^{(2)}(\Psi_{-t}^x,\Psi_{-t}^v)\diff t
	\diff V\diff X.
	\end{gather}
	
	\subsection{The space derivative term}\label{subSec:SpaceDerivative}
	We write $J(\Phi)$ as a sum of time integrals, where for each integral the argument is smooth:
	\begin{align}
	J(\Phi) = \int_{\mathcal{P}}\int_{\mathcal{V}(X)}\sum_{i=-M+1}^{N}\int_{-\tau_i}^{-\tau_{i-1}}
	f_0^{(2)}(\Psi_{-t}^x,\Psi_{-t}^v)V\cdot\nabla_X\Phi(Z,t)\diff t\diff V \diff X,
	\end{align}
	so that on each interval we can apply the following rule:
	\begin{align}
	f_0^{(2)}V\cdot\nabla_X\Phi = \mathrm{div}(Vf_0^{(2)}\Phi) -\Phi V\cdot\nabla_Xf_0^{(2)}.
	\end{align}
	Then we have
	\begin{align}
	J(\Phi) = 
	\int_{\mathcal{P}}\int_{\mathcal{V}(X)}\Bigg\{
	\sum_{i=-M+1}^N\int_{-\tau_i}^{-\tau_{i+1}}\mathrm{div}(Vf_0^{(2)}\Phi)\diff t
	-
	\int_{-\infty}^{\infty}\Phi(X,V,t)V\cdot \nabla_X f^{(2)}(X,V,t)\diff t
	\Bigg\}.
	\end{align}
	We then consider the first (divergence) term in $J(\Phi)$. We split this integral by the velocities $V\in\mathcal{C}(X)$ that cause interactions and $V\in\mathcal{V}(X)\backslash\mathcal{C}(X)$ that do not. For the latter case we have
	\begin{align}
	M_1(\Phi):=  \int_{\mathcal{P}}\int_{\mathcal{V}(X)\backslash\mathcal{C}(X)}\int_{-\infty}^{\infty}\mathrm{div}(Vf_0^{(2)}(\Phi_{-t}^x,\Phi_{-t}^v)\Phi(X,V,t))\diff t
	\end{align}
	which then, by the divergence theorem
	\begin{align}
	M_1(\Phi) =& \int_{\partial\mathcal{P}}\int_{\mathcal{V}(X)\backslash \mathcal{C}(X)}\int_{-\infty}^{\infty}\Phi(X,V,t)f^{(2)}(X,V,t)V\cdot\hat{\nu}(X,V)\diff t\diff V\diff \mathcal{H}(X),
	\end{align}
	where $\hat{\nu}(X,V)$ is the outward unit normal of $\partial_\mathcal{P}$.
	
	For the collisional integrals, we use Reynold's transport theorem \cite{Marsden2003} on each term in the sum to find
	\begin{align}
	\int_{\mathcal{P}}\int_{\mathcal{C}(X)}&
	\sum_{i=-M+1}^N\int_{-\tau_i}^{-\tau_{i+1}}\mathrm{div}(Vf_0^{(2)}\Phi)\diff t \diff V \diff X\nonumber\\
	=&
	\int_{\partial\mathcal{P}}\int_{\mathcal{C}(X)}\sum_{i=-M+1}^N\mathrm{div}_X
	\int_{-\tau_i}^{-\tau_{i-1}}
	V\Phi(X,V,t)f_0^{(2)}(\Psi_{-t})\diff t\diff V \diff X\nonumber\\
	&+\int_{\mathcal{P}}\int_{\mathcal{C}(X)}\sum_{i=-M+1}^N(V\cdot\nabla_X\tau_i)\Phi(X,V,\tau_i)[f_0^{(2)}(\Psi_{\tau_i^-})-f_0^{(2)}(\Psi_{\tau_i^+})]\diff V\diff X.
	\end{align}
	Once again, we split the result $J(\Phi)$ into several parts:
	\begin{align}
	J_1(\Phi) =& \int_{\mathcal{P}}\int_{\mathcal{C}(X)}\sum_{i=-M+1}^N(V\cdot\nabla_X\tau_i)\Phi(X,V,\tau_i)[f_0^{(2)}(\Psi_{\tau_i^-})-f_0^{(2)}(\Psi_{\tau_i^+})]\diff V\diff X\\
	J_2(\Phi) =& -\int_{\mathcal{P}}\int_{\mathcal{V}(X)}\int_{-\infty}^{\infty}
	\Phi(X,V,t)V\cdot\nabla_Xf^{(2)}(X,V,t)\diff t\diff V\diff X,\\
	J_3(\Phi) =& \int_{\partial\mathcal{P}}\int_{\mathcal{V}(X)\backslash\mathcal{C}(X)}\int_{-\infty}^{\infty}\Phi(X,V,t)f^{(2)}(X,V,t)V\cdot\hat{\nu}(X,V)\diff t\diff V\diff \mathcal{H}(X)\nonumber\\
	&+ \int_{\partial\mathcal{P}}\int_{\mathcal{C}(X)}\sum_{i=-M+1}^N\mathrm{div}_X
	\int_{-\tau_i}^{-\tau_{i-1}}
	V\Phi(X,V,t)f_0^{(2)}(\Psi_{-t})\diff t\diff V \diff X.
	\end{align}
	
	\subsection{The velocity derivative term}\label{subSec:VelocityDerivative}
	The velocity derivative follows a similar argument to the one above; we use the following calculus identity between each two event times, as in the spacial derivative:
	\begin{align}
	f^{(2)}\nabla_V\cdot(G\Phi) =
	\mathrm{div}_V(Gf^{(2)}\Phi) - \Phi G\cdot\nabla_Vf^{(2)}.\label{eq:velExpansion}
	\end{align}
	It remains to consider the first term in \cref{eq:velExpansion}. First we consider the case of no interactions. As we have assumed that for each $X\in\mathcal{P}$, $\mathcal{V}(X)\backslash\mathcal{C}(X)$ is a piecewise analytic submanifold of $\mathbb{R}^6$, the divergence theorem provides us with the following result:
	\begin{align}
	M_2(\Phi)
	:=&
	\int_{\mathcal{P}}\int_{\mathcal{V}(X)\backslash\mathcal{C}(X)}\int_{-\infty}^{\infty}\mathrm{div}(G(X,V,t)\Phi(X,V,t)f_0^{(2)}(\Phi_{-t}))\diff t \diff V\diff X,\nonumber\\
	=&
	\int_{\mathcal{P}}
	\int_{\partial(\mathcal{V}(X)\backslash \mathcal{C}(X))}
	\int_{-\infty}^{\infty}
	f^{(2)}(X,V,t)\Phi(X,V,t)
	G\cdot\hat{\eta}_V(V)\diff t\diff\mathcal{H}(V)\diff X.
	\end{align}
	For the collisional part, we have that, using the Reynolds transport theorem in an analogous fashion to before, under the assumption that $\mathcal{C}(X)$ is an analytic submanifold of $\mathcal{V}(X)$:
	\begin{align}
	\int_{\mathcal{P}}\int_{\mathcal{C}(X)}&\sum_{i=-M+1}^N\int_{-\tau_i}^{\tau_i}\mathrm{div}_V(G\Phi f_0^{(2)}(\Psi_{-t}))\diff t \diff V \diff X\nonumber\\
	=&
	\int_\mathcal{P}\int_{\partial C(X)}
	\sum_{i=1}^{N-1}
	\mathrm{div}_V
	\int_{-\tau_i}^{-\tau_{i+1}}
	Gf^{(2)}(Z,t)\Phi(Z,t)\diff t\diff V\diff X\nonumber\\
	+&
	\int_\mathcal{P}\int_{C(X)}
	\sum_{i=1}^{N-1}
	[G\cdot\nabla_V\tau_{i}]\Phi(Z,\tau_{i})
	[f_0^{(2)}(\Psi_{\tau_i^-})-f_0^{(2)}(\Psi_{\tau_i^+})]\diff V\diff X.
	\end{align}
	We partition the result as follows:
	\begin{align}
	K_1(\Phi) =& \int_\mathcal{P}\int_{C(X)}
	\sum_{i=1}^{N-1}
	[G\cdot\nabla_V\tau_{i}]\Phi(Z,\tau_{i})
	[f_0^{(2)}(\Psi_{\tau_i^-})-f_0^{(2)}(\Psi_{\tau_i^+})]\diff V\diff X\\
	K_2(\Phi) =& -\int_{\mathcal{P}}\int_{\mathcal{V}(X)}\int_{-\infty}^{\infty}
	\Phi(X,V,t)G\cdot\nabla_Vf^{(2)}(X,V,t)\diff t\diff V\diff X,\\
	K_3(\Phi) =& \int_{\mathcal{P}}
	\int_{\partial(\mathcal{V}(X)\backslash\mathcal{C}(X))}
	\int_{-\infty}^{\infty}
	f^{(2)}(X,V,t)\Phi(X,V,t)
	G\cdot\hat{\eta}_V(V)\diff t\diff\mathcal{H}(V)\diff X\nonumber\\
	&+\int_\mathcal{P}\int_{\partial C(X)}
	\sum_{i=1}^{N-1}
	\mathrm{div}_V
	\int_{-\tau_i}^{-\tau_{i+1}}
	Gf^{(2)}(Z,t)\Phi(Z,t)\diff t\diff V\diff X.
	\end{align}
	\subsection{Combining all terms}\label{subSec:AllTerms}
	We now combine all contributions into one equation. We note that, by \cref{prop:TIdI},
	\begin{align}
	I_1(\Phi) = -(J_1(\Phi)+K_1(\Phi)),
	\end{align}
	and so when combining all contributions, these terms disappear. 
	By an application of generator theory on $\Psi_{-t}^x,\Psi_{-t}^v$ \cite{Pavliotis2008}, and the results of \Cref{prop:TIdI,prop:TIdII} we see that
	\begin{align}
	I_2(\Phi) = -(J_2(\Phi)+K_2(\Phi)).
	\end{align}
	The remaining terms $J_3(\Phi)$ and $K_3(\Phi)$ are surface terms in position and velocity phase space respectively. By applying the dominated convergence theorem \cite{Royden1988}, and the divergence theorem, we find
	\begin{align}
	J_3(\Phi) =& \int_{\partial\mathcal{P}}\int_{\mathcal{V}(X)}\int_{-\infty}^{\infty}f^{(2)}(X,V,t)\Phi(X,V,t)V\cdot\hat{\nu}\diff t\diff V\diff \mathcal{H}(X),\\
	K_3(\Phi) =& \int_{\mathcal{P}}\int_{\partial\mathcal{V}(X)}\int_{-\infty}^{\infty}f^{(2)}(X,V,t)\Phi(X,V,t)G(X,V,t)\cdot\hat{\nu}_V(X,V)\diff t\diff \mathcal{H}(X,V)\diff X.   
	\end{align}
	This concludes the proof of \cref{Theorem:main}.
	
	As previously noted, the results of \cref{Theorem:main} are quite general. To provide more of an insight into our methodology, we now consider an example of particle dynamics and interparticle interactions that is of interest.
	\section{Dynamics with friction, gravity and inelasticity}\label{subSec:InelasticExample}
	\subsection{Free dynamics}
	For some $G=[g,g]\in\mathbb{R}^6$ where $g\in\mathbb{R}^3$ and $\gamma>0$, we consider the following partial differential equations:
	\begin{align}
	\frac{\diff \Phi_t^x(X,V)}{\diff t} = \phi_t^v(X,V),\quad \frac{\diff \Phi_t^v(X,V)}{\diff t} = -\gamma V - \bm{G}.
	\end{align}
	Physically, these equations are used to model viscous drag and gravitational force. The resulting equations of motion for free particles are well known and can be easily derived:
	\begin{align}
	\Phi_t^x(X,V) =& X -\frac{t}{\gamma}G + \frac{1}{\gamma}\left(
	V + \frac{1}{\gamma}G
	\right)(1-e^{-\gamma t}),\label{eq:FGIPosition}\\
	\Phi_t^v(X,V) =& -\frac{1}{\gamma}G + \left(V+\frac{1}{\gamma}G\right)e^{-\gamma t}.\label{eq:FGIVelocity}
	\end{align}
	
	\subsection{Inelastic Collisions}
	We consider these dynamics where collisions between particles are {\em inelastic}, so that the energy of the system is reduced when two particles collide. To determine the associated scattering map, we include the following additional condition on the forward and reverse event maps $\sigma^\pm(X,V)$:
	\begin{align}
	|\nabla_V\sigma^+(X,V)V| = -\alpha, \quad |\nabla_V\sigma^-(X,V)V| = -\frac{1}{\alpha}\label{eq:MongeAmpere}
	\end{align}
	
	We give the following result on the form of the scattering map.
	\begin{lemma}\label{lemma:InelasticBoltzmann}
		Under the conditions \cref{eq:COLM,eq:COAM,eq:MongeAmpere}, under the additional assumption that $\eta$ in \Cref{Theorem:sigma} is constant, the event maps $\sigma^\pm(X,V)$ have the form
		\begin{align}
		\sigma^+(X,V) = I - (1+\alpha)N(X)\otimes N(X)\nonumber,\\
		\sigma^-(X,V) = I - \frac{1+\alpha}{\alpha}N(X)\otimes N(X).
		\end{align}
	\end{lemma}
	\begin{proof}
		From \Cref{Theorem:sigma}, we know that
		\begin{align*}
		\sigma^\pm(X,V) = I - \eta N(X)\otimes N(X).
		\end{align*}
		Then
		\begin{align*}
		\nabla_V(\sigma(X,V)V) = I - \eta N(X) \otimes N(X).
		\end{align*}
		It remains to solve \cref{eq:MongeAmpere}. As the determinant of a matrix is the product of its eigenvalues $\lambda_i$, we have that
		\begin{align*}
		\prod_{i=1}^6 \lambda_i^+ = -\alpha,\quad \prod_{i=1}^6 \lambda_i^-= -\frac{1}{\alpha}.
		\end{align*}
		For both maps $\lambda_i^\pm=1$ for $i=1,...,5$, by using the 5 independent eigenvectors which are perpendicular to $N(X)$. For $\sigma^+(X,V)$, the remaining eigenvalue, given by a vector parallel to $N(X)$, must be $-\alpha$. Thus,
		\begin{align*}
		V + \eta(X,V)V = -\alpha V.
		\end{align*}
		Rearranging we find $\eta V = (1+\alpha)V$, which gives the required result. The result for $\sigma^-(X,V)$ is analogous.
	\end{proof}
	We remark that upon relaxing the assumption that $\eta(X,V)$ is constant, the (first) Monge-Ampere equation becomes
	\begin{align}
	|I - N(X)\cdot VN(X)\otimes\nabla_V(\eta^+(X,V)) -\eta^+(X,V)N(X)\otimes N(X)|=-\alpha.\label{eq:MongeAmpereGeneral}
	\end{align}
	In particular, this could result in physically valid {\em non-linear} scattering maps $\sigma^\pm(X,V)$ for a particular event. Understanding this equation is an interesting topic for future work. In this section we focus on the inelastic Boltzmann scattering maps defined in \Cref{lemma:InelasticBoltzmann} and note that when $\alpha=1$ these reduce to the elastic Boltzmann scattering map considered in \cite{Wilkinson2018}. 
	
	\subsection{Velocity cones and collision times}
	As the reduced difference dynamics $\tilde{\Phi}_t^x(X,V)$ follow straight lines (parametrised exponentially in $-\gamma t$), we see that two particles satisfying \cref{eq:FGIPosition,eq:FGIVelocity} can experience at most one collision. The initial data can be partitioned into non-interacting, pre-collisional and post-collisional.
	
	For $X\in\mathcal{P}$ and $V\in\mathbb{R}^6$, we write $L(X,V)\subset\mathbb{R}^6$ to denote the line
	\begin{align}
	L(X,V) = \left\{
	X -\frac{gt}{\gamma} + \frac{1}{\gamma}\left(
	V + \frac{g}{\gamma}
	\right)(1-e^{-\gamma t})
	: t\in\mathbb{R}
	\right\}
	\end{align}
	and define the two infinite half lines
	\begin{align}
	L^-(X,V) = \left\{
	X -\frac{gt}{\gamma} + \frac{1}{\gamma}\left(
	V + \frac{g}{\gamma}
	\right)(1-e^{-\gamma t})
	: t\le 0
	\right\},\\
	L^+(X,V) = \left\{
	X -\frac{gt}{\gamma} + \frac{1}{\gamma}\left(
	V + \frac{g}{\gamma}
	\right)(1-e^{-\gamma t})
	: t\ge 0
	\right\}
	\end{align}
	The velocity collision cone is then defined analogously
	\begin{align}
	C(X) = \left\{
	V\in\mathbb{R}^6: L(X,V)\cap\partial\mathcal{P}\neq\emptyset
	\right\}
	\end{align}
	And we can split this set into precollisional and postcollisional velocities respectively:
	\begin{align}
	C^-(X) = \left\{
	V\in C(X): L^+(X,V)\cap\partial\mathcal{P}\neq\emptyset
	\right\},\\
	C^+(X) = \left\{
	V\in C(X): L^-(X,V)\cap\partial\mathcal{P}\neq\emptyset
	\right\}.
	\end{align}
	
	We note that a constant external potential $G$ does not have an effect on the shape of the collision cones, as it does not effect the dynamics determined by the relative distance of the particles. The frictional constant $\gamma$ truncates the precollisional velocity cone (when compared to linear dynamics). 
	However, for these dynamics, for any given $X\in\mathcal{P}$, all initial velocities $V\in\mathbb{R}^6$ are admissible, and so in particular the second surface term in the weak formulation of the Liouville equation disappears.

	
	
	In this case we can analytically construct the unique event time $\tau(X,V)$.
	
	\begin{lemma}[Characterisation of the Collision Time Map for dynamics with gravity and friction]
		For any $X\in\mathcal{P}$,
		
		\begin{enumerate}
			\item If $V\in\mathbb{C}^+(X)$ then $\tilde{x}\cdot\tilde{v}>0$ and 
			\begin{align}
			\tau(X,V) =
			-\frac{1}{\gamma}\log\left(
			1+\frac{\gamma}{\|\tilde{v}\|}\left\{
			\tilde{x}\cdot\widehat{\tilde{v}}
			+
			\left[
			(\tilde{x}\cdot\widehat{\tilde{v}})^2
			-
			(\|\tilde{x}\|^2-\varepsilon^2)
			\right]^{\frac{1}{2}}
			\right\}
			\right).
			\end{align}
			\item If $V\in C^-(X)$ then $-\frac{1}{2}\left(
			\gamma(\|x\|^2-\varepsilon^2)+\frac{\|\tilde{v}\|^2}{\gamma}
			\right)<\tilde{x}\cdot\tilde{v}<0$ and
			\begin{align}
			\tau(X,V) = 
			-\frac{1}{\gamma}\log\left(
			1+\frac{\gamma}{\|\tilde{v}\|}\left\{
			\tilde{x}\cdot\widehat{\tilde{v}}
			-
			\left[
			(\tilde{x}\cdot\widehat{\tilde{v}})^2
			-
			(\|\tilde{x}\|^2-\varepsilon^2)
			\right]^{\frac{1}{2}}
			\right\}
			\right).
			\end{align}
		\end{enumerate}
	\end{lemma}
	\begin{proof}
		The collision occurs when
		\begin{align*}
		\|\tilde{\Phi}_t^x(X,V)\|^2 = \varepsilon^2
		\implies
		\|\tilde{x}+\frac{1}{\gamma}(\tilde{v})(1-e^{-\gamma \tau})\|=\varepsilon^2.
		\end{align*}
		By expanding the left hand side and rearranging, we have that if the particles collide, then $\tau(X,V)$ must take one of the two following values
		\begin{align*}
		\tau^-(X,V)=&
		-\frac{1}{\gamma}\log\left(
		1+\frac{\gamma}{\|\tilde{v}\|}\left\{
		\tilde{x}\cdot\widehat{\tilde{v}}
		-
		\left[
		(\tilde{x}\cdot\widehat{\tilde{v}})^2
		-
		(\|\tilde{x}\|^2-\varepsilon^2)
		\right]^{\frac{1}{2}}
		\right\}
		\right),\\
		\tau^+(X,V)=&-\frac{1}{\gamma}\log\left(
		1+\frac{\gamma}{\|\tilde{v}\|}\left\{
		\tilde{x}\cdot\widehat{\tilde{v}}
		+
		\left[
		(\tilde{x}\cdot\widehat{\tilde{v}})^2
		-
		(\|\tilde{x}\|^2-\varepsilon^2)
		\right]^{\frac{1}{2}}
		\right\}
		\right).
		\end{align*}
		We note that the argument in the square root requires $\tilde{x}\cdot\hat{\tilde{v}}<-(\|\tilde{x}\|^2-\varepsilon^2)^{\frac{1}{2}}$ or  $\tilde{x}\cdot\hat{\tilde{v}}>(\|\tilde{x}\|^2-\varepsilon^2)^{\frac{1}{2}}$, and if both $\tau^\pm(X,V)$ exist, then $\tau^-(X,V)<\tau^+(X,V)$. Assume now that $V\in C^+(X)$, then $\tau^\pm(X,V)<0$, and so $\tau^+(X,V)>\tau^-(X,V)$, so $\tau(X,V)=\tau^+(X,V)$ is required. Furthermore,
		\begin{align*}
		\tilde{x}\cdot\tilde{v}>\frac{\|\tilde{v}\|}{\gamma} + ((\tilde{x}\cdot\tilde{v})^2-(\|\tilde{x}\|^2-\varepsilon^2))^{\frac{1}{2}}>0.
		\end{align*}
		and so $\tilde{x}\cdot\tilde{v}>0$. Alternatively, if $V\in C^-(X)$, $\tau(X,V)>0$ and so we take $\tau(X,V)=\tau^-(X,V)$. Thus
		\begin{align*}
		\tilde{x}\cdot\hat{\tilde{v}}+\frac{\|\tilde{v}\|}{\gamma}
		>((\tilde{x}\cdot\hat{\tilde{v}})^2-(\|\tilde{x}\|^2-\varepsilon^2))^{\frac{1}{2}},
		\end{align*}
		Squaring both sides and rearranging, we find
		\begin{align*}
		\tilde{x}\cdot\tilde{v}>-\frac{1}{2}\left(\gamma(\|\tilde{x}\|^2-\varepsilon^2)+\frac{\|\tilde{v}\|}{\gamma}\right).
		\end{align*}
	\end{proof}
	The inequality in the pre-collisional case relates to the presence of friction in the dynamics; if particles do not have enough energy in the direction $\tilde{x}$ then the particles will never meet.
	
	\subsection{Flow maps}
	Given the free particle dynamics, the scattering maps, and a full characterisation of the admissible data, we are now in a position to define the hard sphere flow maps $T_t$. We split into collision free and collisional dynamics.
	
	\paragraph{Collision free Dynamics} If $(X,V)\in\{X\}\times\mathbb{R}^6\backslash\mathcal{C}(X)$, then
	\begin{align}
	(\Pi_1\otimes T_t)Z =&  X -\frac{gt}{\gamma} + \frac{1}{\gamma}\left(
	V + \frac{g}{\gamma}
	\right)(1-e^{-\gamma t})\\
	(\Pi_2\otimes T_t)Z =& -\frac{g}{\gamma} + \left(V+\frac{g}{\gamma}\right)e^{-\gamma t}.
	\end{align}
	
	\paragraph{Collisional Dynamics} Firstly, if $(X,V)\in \{X\}\times C^-(X)$, then
	\begin{align}
	(\Pi_1\otimes T_t)Z =&  
	\begin{cases}
	X -\frac{gt}{\gamma} + \frac{1}{\gamma}\left(
	V + \frac{g}{\gamma}
	\right)(1-e^{-\gamma t})
	, \quad\text{ if } -\infty<t<\tau(X,V)\\
	\begin{aligned}[c]&
	\left[
	X -\frac{g\tau}{\gamma} + \frac{1}{\gamma}\left(
	V + \frac{g}{\gamma}
	\right)(1-e^{-\gamma \tau})
	\right]
	-\frac{g(t-\tau)}{\gamma}\\&\qquad
	+\frac{1}{\gamma}
	\left[
	\sigma_-(X,V)
	\left(
	(V+\frac{g}{\gamma})e^{-\gamma\tau}-\frac{g}{\gamma}
	\right)
	+\frac{g}{\gamma}
	\right](1-e^{-\gamma(t-\tau)})
	\end{aligned}
	, \quad\text{ if } \tau(X,V)<t<\infty\\
	\end{cases}
	\end{align}
	and
	\begin{align}
	(\Pi_2\otimes T_t)Z =&  
	\begin{cases}
	-\frac{g}{\gamma} + \left(V+\frac{g}{\gamma}\right)e^{-\gamma t}
	, \quad\text{ if } -\infty<t<\tau(X,V)\\
	-\frac{g}{\gamma}+\left[
	\sigma_-(X,V)\left(
	(V+\frac{g}{\gamma})e^{-\gamma\tau}-\frac{g}{\gamma}
	\right)
	+\frac{g}{\gamma}
	\right]
	e^{-\gamma(t-\tau)}
	, \quad\text{ if } \tau(X,V)<t<\infty\\
	\end{cases}
	\end{align}
	If $(X,V)\in \{X\}\times C^+(X)$, then
	\begin{align}
	(\Pi_1\otimes T_t)Z =&  
	\begin{cases}
	\begin{aligned}[c]&
	\left[
	X -\frac{g\tau}{\gamma} + \frac{1}{\gamma}\left(
	V + \frac{g}{\gamma}
	\right)(1-e^{-\gamma \tau})
	\right]
	-\frac{g(t-\tau)}{\gamma}\\&\qquad
	+\frac{1}{\gamma}
	\left[
	\sigma_+(X,V)
	\left(
	(V+\frac{g}{\gamma})e^{-\gamma\tau}-\frac{g}{\gamma}
	\right)
	+\frac{g}{\gamma}
	\right](1-e^{-\gamma(t-\tau)})
	\end{aligned}
	, \quad\text{ if } -\infty<t<\tau(X,V)\\
	X -\frac{gt}{\gamma} + \frac{1}{\gamma}\left(
	V + \frac{g}{\gamma}
	\right)(1-e^{-\gamma t})
	, \quad\text{ if } \tau(X,V)<t<\infty\\
	\end{cases}
	\end{align}
	and
	\begin{align}
	(\Pi_2\otimes T_t)Z =&  
	\begin{cases}
	-\frac{g}{\gamma}+\left[
	\sigma_+(X,V)\left(
	(V+\frac{g}{\gamma})e^{-\gamma\tau}-\frac{g}{\gamma}
	\right)
	+\frac{g}{\gamma}
	\right]
	e^{-\gamma(t-\tau)}
	, \quad\text{ if }-\infty<t<\tau(X,V)\\
	-\frac{g}{\gamma} + \left(V+\frac{g}{\gamma}\right)e^{-\gamma t}
	, \quad\text{ if } \tau(X,V)<t<\infty\\
	\end{cases}
	\end{align}
	Finally, if $X\in\partial\mathcal{P}$
	\begin{align}
	(\Pi_1\otimes T_t)Z =&  
	\begin{cases}
	X -\frac{gt}{\gamma} + \frac{1}{\gamma}\left(
	V + \frac{g}{\gamma}
	\right)(1-e^{-\gamma t})
	, \quad\text{ if } -\infty<t<0\\
	X -\frac{gt}{\gamma} + \frac{1}{\gamma}\left(
	\sigma_-(X,V)V + \frac{g}{\gamma}
	\right)(1-e^{-\gamma t})
	, \quad\text{ if } 0<t<\infty\\
	\end{cases}
	\end{align}
	and
	\begin{align}
	(\Pi_2\otimes T_t)Z =&  
	\begin{cases}
	-\frac{g}{\gamma} + \left(V+\frac{g}{\gamma}\right)e^{-\gamma t}
	, \quad\text{ if } -\infty<t<0\\
	-\frac{g}{\gamma} + \left(\sigma_-(X,V)V+\frac{g}{\gamma}\right)e^{-\gamma t}
	, \quad\text{ if } 0<t<\infty.\\
	\end{cases}
	\end{align}
	
	\subsection{The Liouville equation}
	Using \cref{Theorem:main}, as we understand the admissible data $X,V$ for the particle dynamics, we can write down the Liouville equation for these particular dynamics. For any $f_0^{(2)}\in C^0(\mathcal{D})\cap L^1(\mathcal{D})$, there exists a physical global in time weak solution of
	\begin{align}
	\left[\frac{\partial}{\partial t}+V\cdot\nabla_X - G\cdot\nabla_V -\nabla_V\cdot(\gamma V)\right]f^{(2)}(X,V,t)
	=
	C_X[f^{(2)}],
	\end{align}
	where the scattering map defining collisions satisfies \cref{eq:MongeAmpere}.
	
	\subsection{The BBGKY hierarchy}
	We define $f^{(1)}(x,v,t)$ as in \Cref{Cor:BBGKY}, which then satisfies (by using $G=[g,g]$),
	\begin{align*}
	&\int_{\mathcal{P}}\int_{\mathbb{R}^6}\int_{-\infty}^{\infty}
	(\partial_t + v\cdot\nabla_x +g\cdot\nabla_v)\phi(x,v,t)f^{(1)}(x,v,t)\diff t\diff V\diff X\\
	=&
	-\frac{1}{\sqrt{2}}\int_{\mathbb{R}^3}\int_{\mathbb{S}^2}\int_{\mathbb{R}^6}\int_{-\infty}^{\infty}
	\phi(x,v,t)f^{(2)}([x,x+\varepsilon n],[v,\bar{v}],t)(v-\bar{v})\cdot n\diff t \diff V\diff n \diff x.
	\end{align*}
	The collisional term $C[f^{(2)}]$ in the above equation can then be separated into a pre-collisional and a post-collisional term:
	\begin{gather*}
	C[f^{(2)}] = \frac{1}{\sqrt{2}}\int_{\mathbb{R}^3}\int_{\mathbb{S}^2}\int_{C^-(X)}\int_{-\infty}^{\infty}
	\phi(x,v,t)f^{(2)}([x,x+\varepsilon n],[v,\bar{v}],t)(v-\bar{v})\cdot n\diff t \diff V\diff n \diff x\\
	\frac{1}{\sqrt{2}}\int_{\mathbb{R}^3}\int_{\mathbb{S}^2}\int_{C^+(X)}\int_{-\infty}^{\infty}
	\phi(x,v,t)f^{(2)}([x,x+\varepsilon n],[v,\bar{v}],t)(v-\bar{v})\cdot n\diff t \diff V\diff n \diff x,\\
	\end{gather*}
	where
	\begin{gather*}
	C^+(n) = \{ V\in\mathbb{R}^6:(v-\bar{v})\cdot n >0\}
	\end{gather*}
	and
	\begin{align*}
	C^-(n) =& \{ V=[v,\bar{v}]\in\mathbb{R}^6:-\frac{\|v-\bar{v}\|^2}{2\gamma}<(v-\bar{v})\cdot n <0\}\\
	=& \{ V=[v,\bar{v}]\in\mathbb{R}^6:(v-\bar{v})\cdot n <0\},\\
	\end{align*}
	where the lower bound has disappeared because we are considering $X\in\partial\mathcal{P}$. We introduce the change of variables for the post-collisional integral that is motivated by the backward time scattering map:
	\begin{align*}
	V\mapsto \left( 
	I - \frac{1+\alpha}{2\alpha}\begin{Bmatrix}
	n \\ -n
	\end{Bmatrix}
	\otimes
	\begin{Bmatrix}
	n \\ -n
	\end{Bmatrix}
	\right)V.
	\end{align*}
	This transform has Jacobian $-1/\alpha$. We note that
	\begin{align*}
	(v'_n-\bar{v}_n')\cdot n = -\frac{1}{\alpha}(v-\bar{v})\cdot n,
	\end{align*}
	where the primed values are determined by the backward time Boltzmann inelastic scattering map, which in terms of $v,\bar{v}$ is the inverse inelastic collision rule
	\begin{align}
	v'_n       = v-\frac{1+\alpha}{2\alpha}(n\cdot(v-\bar{v}))n,\\
	\bar{v}_n' = v+\frac{1+\alpha}{2\alpha}(n\cdot(v-\bar{v}))n,
	\end{align}
	and so we obtain the inelastic collision operator in the first equation of the BBGKY hierarchy:
	\begin{align}
	&\int_{\partial\mathcal{P}}\int_{\mathbb{R}^6}\int_{-\infty}^{\infty}\Phi(X,V,t)F(X,V,t) V\cdot\tilde{\nu}(Y)dtdVdY,\\
	=&\frac{1}{\sqrt{2}}
	\int_{\mathbb{R}^3}\int_{\mathbb{S}^2}\int_{C^-(n)}\int_{-\infty}^{\infty}\Phi(x,v,y)\Bigg[
	F^{(2)}(y,v,y+\varepsilon n,\bar{v},t)\\&\qquad \qquad \qquad
	-\frac{1}{\alpha^2}F^{(2)}(y,v_n',y+\varepsilon n,\bar{v}'_n,t)
	\Bigg](v-\bar{v})\cdot n \diff t\diff V\diff n \diff x.
	\end{align}
	We note that, upon additional assumptions ({\em i.e.} molecular chaos), this is a weak analogue of the inelastic Boltzmann collision operator.
	\section{Conclusions and future work}\label{Sec:Conclusions}
	
	In this paper we have presented a weak formulation for particles under the influence of a general dynamical form, with general instantaneous interactions. We have an example which is consistent with results in the literature on inelastic collision operators in the BBGKY hierarchy, and generalizes the Liouville equation. This shows the potential of the results for more complicated particle systems.
	
	In the future, we wish to extend these results to include interactions modelled by general step potentials, and investigate the entire space of admissible initial data. From here we can consider homogenisation procedures to construct an effective potential term for particles with interactions based on step potentials, which can then be incorporated into a modern computational model.
	
	The results presented could also be extended to systems with additional degrees of freedom, for example rotation or `colour'. Furthermore, by investigating systems of many particles we aim to see what initial configurations can be modelled by the Liouville equation presented here. To do so, we need to consider many body interactions and the effect of inelastic collapse. 
	
	\section{Acknowledgements}
	TH was supported by The Maxwell Institute Graduate School in Analysis and its
	Applications, a Centre for Doctoral Training funded by the EPSRC (EP/L016508/01), the Scottish Funding Council, Heriot-Watt University and the University of Edinburgh.  MW is supported by the EPSRC Standard Grant EP/P011543/1. BDG would like to acknowledge support from EPSRC EP/L025159/1.
	
	\appendix	
	\section{Tensor Notation}
	In this paper, we also use tensor product in several calculations. Given $A=(a_{ij})_{i=1,...,n, j=1,...,m}\in\mathbb{R}^{n\times m}$ and $B=(b_{ij})_{i=1,...,k, j=1,...,l}\in\mathbb{R}^{k\times l}$, the tensor product $A\otimes B\in\mathbb{R}^{nk\times ml}$ is given by the following element-wise multiplication:
	\begin{align}
	A\otimes B = 
	\begin{pmatrix}
	a_{11}b_{11} & a_{11}b_{12} & \dots  & a_{11}b_{1l} & \dots & \dots & a_{1m}b_{11} & a_{1m}b_{12} & \dots  & a_{1m}b_{1l} \\
	a_{11}b_{21} & a_{11}b_{22} & \dots  & a_{11}b_{2l} & \dots & \dots & a_{1m}b_{21} & a_{1m}b_{22} & \dots  & a_{1m}b_{2l} \\
	\vdots       & \vdots       & \ddots & \vdots       &       &       &
	\vdots       & \vdots       & \ddots & \vdots       \\
	a_{11}b_{k1} & a_{11}b_{k2} & \dots  & a_{11}b_{kl} & \dots & \dots & a_{1m}b_{k1} & a_{1m}b_{k2} & \dots  & a_{1m}b_{kl} \\
	\vdots       & \vdots       &        & \vdots       & \ddots&       &
	\vdots       & \vdots       &        & \vdots       \\
	\vdots       & \vdots       &        & \vdots       &       & \ddots&
	\vdots       & \vdots       &        & \vdots       \\
	a_{n1}b_{11} & a_{n1}b_{12} & \dots  & a_{n1}b_{1l} & \dots & \dots & a_{nm}b_{11} & a_{nm}b_{12} & \dots  & a_{nm}b_{1l} \\
	a_{n1}b_{21} & a_{n1}b_{22} & \dots  & a_{n1}b_{2l} & \dots & \dots & a_{nm}b_{21} & a_{nm}b_{22} & \dots  & a_{nm}b_{2l} \\
	\vdots       & \vdots       & \ddots & \vdots       &       &       &
	\vdots       & \vdots       & \ddots & \vdots       \\
	a_{n1}b_{k1} & a_{n1}b_{k2} & \dots  & a_{n1}b_{kl} & \dots & \dots & a_{nm}b_{k1} & a_{nm}b_{k2} & \dots  & a_{nm}b_{kl} \\
	\end{pmatrix}
	\end{align}
	We extend the tensor products to vectors $A=(a_i)_{i=1,...,n, }\in\mathbb{R}^{n}$ and $B=(b_i)_{i=1,...,m}\in\mathbb{R}^{m}$ as $A\otimes B\in\mathbb{R}^{n\times m}$ where
	\begin{align}
	A\otimes B = 
	\begin{pmatrix}
	a_1b_1 & a_1b_2 & \dots & a_1b_m \\
	a_2b_1 & a_2b_2 & \dots & a_2b_m \\
	\vdots & \vdots & \ddots& \vdots \\
	a_nb_1 & a_nb_2 & \dots & a_nb_m
	\end{pmatrix}
	\end{align}
	Tensor product notation is particularly useful when considering multidimensional derivatives of vectors in this paper; for $A\in\mathbb{R}^n$ we and a differentiable function $F:A\rightarrow\mathbb{R}^m$,we define $\nabla_AF(A)\in\mathbb{R}^{n\times m}$ as
	\begin{align}
	\nabla_AF(A) =
	\begin{pmatrix}
	\partial_{a_1} & \partial_{a_2} & \dots & \partial_{a_n}
	\end{pmatrix}
	\otimes
	F(A)
	\end{align}
	For example,
	\begin{align}
	\nabla_X\tilde{x} = 
	\begin{pmatrix}
	1 & 0 & 0 \\
	0 & 1 & 0 \\
	0 & 0 & 1 \\
	-1 & 0  & 0 \\   
	0 & -1 & 0 \\
	0 &  0 & -1
	\end{pmatrix}
	=  \begin{pmatrix} 1 \\ -1 \end{pmatrix} \otimes I_3
	\end{align}
	where $I_3\in\mathbb{R}^{3\times3}$ is the identity matrix. Derivations also involve matrix vector products, and to avoid ambiguity, for a vector $a\in\mathbb{R}^n$ and matrix $b\in\mathbb{R}^{n\times n}$, we define $a\cdot B$ and $B\cdot a$ elementwise by
	\begin{align}
	(a\cdot B)_i =\sum_{j=1}^n a_j B_{ji},\nonumber\\
	(B\cdot a)_i =\sum_{j=1}^n a_j B_{ij},
	\end{align}
	for $i=1,...,n$.
	
	\section{One dimensional solutions to the Monge Ampere equations}
	In \cref{subSec:InelasticExample} we introduced an event map constraint in the form of a Monge-Ampere equation \cref{eq:MongeAmpereGeneral}. Under the additional assumption that $\eta(X,V)$ is a constant, we find $\eta^+(X,V) = -(1+\alpha)$ and $\eta^-(X,V)=-\frac{(1+\alpha)}{\alpha}$. It is unclear whether other (non-linear) solutions to \cref{eq:MongeAmpereGeneral} exist. In this appendix we produce a non-linear solution in one dimension. Firstly, in one dimension,
	\begin{align*}
	N(X) = \frac{1}{\sqrt{2}}[1,-1],\quad N(X)\otimes N(X) = \frac{1}{2}\begin{pmatrix} 1 & -1 \\ -1 & 1 \end{pmatrix}.
	\end{align*}
	and so
	\begin{align}
	|D\sigma(X,V)V| = 
	\left|
	\begin{pmatrix}
	1 + \frac{1}{2}(v-\bar{v})\partial_v\eta + \frac{\eta}{2} &
	\frac{1}{2}(v-\bar{v})\partial_{\bar{v}}\eta - \frac{\eta}{2} \\
	-\frac{1}{2}(v-\bar{v})\partial_{v}\eta - \frac{\eta}{2} &
	1 + \frac{1}{2}(v-\bar{v})\partial_{\bar{v}}\eta + \frac{\eta}{2} 
	\end{pmatrix}
	\right|.
	\end{align}
	Thus, after some cancellations we have that
	\begin{align}
	2\eta + (v-\bar{v})(\partial_v-\partial_{\bar{v}})\eta = -2(1+\alpha).
	\end{align}
	Here upon assuming that $\eta$ is constant we see that the unique solution is $\eta=-(1+\alpha)$. If we write $\eta(X,V) = \eta(\tilde{v})$, then
	\begin{align}
	\eta + \tilde{v}\partial_{\tilde{v}}\eta = -(1+\alpha).
	\end{align}
	Then for any $c\in\mathbb{R}$,
	\begin{align}
	\eta = \frac{c}{\tilde{v}} - (1+\alpha),
	\end{align}
	is a solution of \cref{eq:MongeAmpereGeneral}. As we require $\sigma^-(X,V)\sigma^+(X,V)=I$, we have that
	\begin{align}
	\eta^+(X,V) = \frac{c}{\tilde{v}} - (1+\alpha),\quad
	\eta^-(X,V) = \frac{\tilde{v}}{c-\alpha\tilde{v}} -1.
	\end{align}
	We note that if $\tilde{v} = c/\alpha$, $\eta^-(X,V)$ is not defined, so these scattering maps can only be applied on a restricted set of initial data. However, existence of a non-linear solution to \cref{eq:MongeAmpereGeneral} encourages further investigation.
	\bibliography{LiouvilleBib}
	\bibliographystyle{ieeetr}	
	
\end{document}